\documentclass[a4paper,UKenglish,cleveref, autoref, thm-restate]{lipics-v2019}

\hypersetup{bookmarksnumbered=true}

\title{The Complexity of Verifying Loop-Free Programs as Differentially Private} %

\author{Marco Gaboardi}{Boston University, USA}{}{}{}%

\author{Kobbi Nissim}{Georgetown University, USA}{}{https://orcid.org/0000-0002-6632-8645}{}

\author{David Purser}{University of Warwick, UK \and MPI-SWS, Germany}{}{https://orcid.org/0000-0003-0394-1634}{}

\authorrunning{M. Gaboardi, K. Nissim and D. Purser} %

\Copyright{Marco Gaboardi, Kobbi Nissim and David Purser} %

\ccsdesc[500]{Security and privacy}
\ccsdesc[500]{Theory of computation~Probabilistic computation}

\keywords{differential privacy, program verification, probabilistic programs} %

\category{Track B: Automata, Logic, Semantics, and Theory of Programming}

\relatedversion{A full version of the paper is available at \url{https://arxiv.org/abs/1911.03272}} %

\supplement{}%

\funding{M.~G.\ and K.~N.\ were supported by NSF grant No.~1565387 TWC: Large: Collaborative: Computing Over Distributed Sensitive Data. D.~P.\ was supported by the UK EPSRC Centre for Doctoral Training in Urban Science (EP/L016400/1).}%

\acknowledgements{Research partially done while M.G.\ and K.N.\ participated in the ``Data Privacy: Foundations and Applications'' program held at the Simons Institute, UC Berkeley in spring 2019.}%

\nolinenumbers %

\EventEditors{Artur Czumaj, Anuj Dawar, and Emanuela Merelli}
\EventNoEds{3}
\EventLongTitle{47th International Colloquium on Automata, Languages, and Programming (ICALP 2020)}
\EventShortTitle{ICALP 2020}
\EventAcronym{ICALP}
\EventYear{2020}
\EventDate{July 8--11, 2020}
\EventLocation{Saarbrücken, Germany (virtual conference)}
\EventLogo{}
\SeriesVolume{168}
\ArticleNo{129}

\clearpage{}%

\usepackage{amsmath}
\usepackage{amssymb}
\usepackage{bbold}

\usepackage{graphicx}
\usepackage[many]{tcolorbox}

\usepackage[bold,full]{complexity}

\usepackage{listings}
\usepackage{xspace}
\usepackage{float}
\usepackage{multirow}

\usepackage{array}

\newcommand{\pr}{\ensuremath{\Pr}}%
\newcommand{\ems}{\textsc{E-Maj-Sat}\xspace}

\newcommand{\ams}{{\normalfont \textsc{All-Min-Sat}}\xspace}

\newcommand{\afs}{{\normalfont \textsc{All-Frac-$f$-Sat}}\xspace}
\newcommand{\efps}[1]{{\normalfont \textsc{E-Frac-#1-Sat}}\xspace}
\newcommand{\efs}{{\normalfont \efps{$f$}}}

\newcommand{\ssat}{\ensuremath{\#\textsc{Sat}}}
\newcommand{\ccsat}{\ensuremath{\#\textsc{CircuitSat}}}

\renewcommand{\epsilon}{\varepsilon}

\newcommand{\decideedp}{\texorpdfstring{{\normalfont \textsc{Decide-$\epsilon$-DP}}\xspace}{Decide epsilon DP}}
\newcommand{\decideeddp}{\texorpdfstring{{\normalfont \textsc{Decide-$\epsilon,\delta$-DP}}\xspace}{Decide epsilon delta DP}}

\newcommand{\conpsp}{\texorpdfstring{\ensuremath{\coNP^{\#\P}}}{coNP\^{}\#P}}
\newcommand{\conpspsp}{\texorpdfstring{\ensuremath{\coNP^{\#\P^{\#\P}}}}{coNP\^{}\#P\^{}\#P}}
\newcommand{\psp}{\ensuremath{\P^{\#\P}}}

\newcommand{\sharpp}{\ensuremath{\#\P}}
\newcommand{\ptime}{\ensuremath{\P}}

\newcommand\zo[1]{\ensuremath{\{0,1\}^{#1}}}

\usepackage{bm}
\newcommand\genalloc[1]{\ensuremath{\bm{{#1}}}}

\newcommand\outpt{\genalloc{o}}
\newcommand\inp{\genalloc{x}}
\newcommand\inpp{\genalloc{x'}}
\newcommand\proballoc{\genalloc{r}}
\newcommand\xa{\genalloc{x}}
\newcommand\ya{\genalloc{y}}

\newcommand\ind[1]{\ensuremath{\mathbb{1}{\left\{#1\right\}}}}

\newcommand\eddpadj[1]{{$(\epsilon#1,\delta#1)$-differentially private}}
\newcommand\eddpnoun[1]{{$(\epsilon#1,\delta#1)$-differential privacy}}

\newcommand\disinguished{\textsc{Distinguish-$(\epsilon,\delta),(\epsilon',\delta')$-DP}\xspace}

\newcounter{SideNoteCounter} \stepcounter{SideNoteCounter}

\newcommand\cutout[1]{}

\newcommand{\generalcircuit}{\ensuremath{\psi}}
\newcommand{\generalformula}{\ensuremath{\phi}}

\newcommand{\remove}[1]{}

\providecommand{\myceil}[1]{\left \lceil #1 \right \rceil }
\providecommand{\bits}[1]{\myceil{\log(v)}}\clearpage{}%

\theoremstyle{definition} 
\newtheorem{reformulation}[theorem]{Reformulation}

\usepackage{bold-extra}

\newcommand{\lipicscenter}[1]{\begin{center}$\displaystyle#1$\end{center}}

\begin{document}

\maketitle

\begin{abstract}
We study the problem of verifying differential privacy for loop-free programs with probabilistic choice. Programs in this class can be seen as randomized Boolean circuits, which we will use as a formal model to answer two different questions: first, deciding whether a  program satisfies a prescribed level of privacy; second, approximating the privacy parameters a program realizes. 
We show that the problem of deciding whether a program satisfies $\epsilon$-differential privacy is $\conpsp{}$-complete. In fact, this is the case when either the input domain or the output range of the program is large. Further, we show that deciding whether a program is $(\epsilon,\delta)$-differentially private is $\conpsp{}$-hard, and in $\conpsp{}$ for small output domains, but always in $\conpspsp$. 
Finally, we show that the problem of approximating the level of differential privacy is both $\NP$-hard and $\coNP$-hard. 
These results complement previous results by Murtagh and Vadhan~\cite{murtagh2016complexity} showing that deciding the optimal composition of differentially private components is $\sharpp$-complete, and that approximating the optimal composition of differentially private components is in $\ptime$. \end{abstract}

\section{Introduction}

Differential privacy~\cite{dwork2006calibrating} is currently making significant strides towards being used in large scale applications. Prominent real-world examples include the use of differentially private computations by the US Census' OnTheMap project\footnote{\url{https://onthemap.ces.census.gov}}, applications by companies such as Google and Apple~\cite{erlingsson2014rappor,papernot2016semi,appleDP,dp2017learning}, and the US Census' plan to deploy differentially private releases in the upcoming 2020 Decennial~\cite{abowd2018us}.

More often than not, algorithms and their implementations are analyzed ``on paper'' to show that they provide differential privacy. This analysis---a proof that the outcome distribution of the algorithm is stable under the change in any single individual's information---is often intricate and may contain errors (see~\cite{LyuSL17} for an illuminating  discussion about several wrong versions of the sparse vector algorithm which appeared in the literature). 
Moreover, even if it is actually differentially private, an algorithm may be incorrectly implemented when used in practice, e.g., due to coding errors, or because the analysis makes assumptions which do not hold in finite computers, such as the ability to sample from continuous distributions (see~\cite{mironov-2012-bits} for a discussion about privacy attacks on naive implementations of continuous distributions). Verification tools may help validate, given the code of an implementation, that it would indeed provide the privacy guarantees it is intended to provide. However, despite the many verification efforts that have targeted differential privacy based on automated or interactive techniques (see, e.g.,\cite{reed2010distance,barthe2012probabilistic,tschantz2011formal,fredrikson2014satisfiability,barthe2016proving,zhang2017lightdp,BGGHRS15,albarghouthi2017synthesizing,chistikov2018bisimilarity,chistikov2019asymmetric}), little is known about the complexity of some of the basic problems in this area. Our aim is to clarify the complexity of some of these problems.

In this paper, we consider the computational complexity of determining whether programs satisfy  \eddpnoun{}. The problem is generally undecidable, and we hence restrict our attention to probabilistic loop-free programs, which are part of any reasonable programming language supporting random computations. To approach this question formally, we consider probabilistic circuits. The latter are Boolean circuits with input nodes corresponding both to input bits and to uniformly random bits (``coin flips'') where the latter allow the circuit to behave probabilistically (see Figure~\ref{fig:egrc}). We consider both decision and approximation versions of the problem, where in the case of decision the input consists of a randomized circuit and parameters $\epsilon,\delta$ and in the case of approximation the input is a randomized circuit, the desired approximation precision, and one of the two parameters $\epsilon,\delta$. In both cases, complexity is measured as function of the total input length in bits (circuit and parameters).

Previous works have studied the complexity of composing differentially private components. 
For any $k$ differentially private algorithms with privacy parameters $(\epsilon_1,\delta_1), \dots, (\epsilon_k,\delta_k)$, it is known that their composition is also differentially private~\cite{dwork2006calibrating,DRV10,murtagh2016complexity}, making composition a powerful design tool for differentially private programs. 
However, not all interesting differentially private programs are obtained by composing differentially private components, and a goal of our work is to understand what is the complexity of verifying that full programs are differentially private, and how this complexity differs from the one for programs which result of composing differentially private components.

Regarding the resulting parameters, the result of composing the $k$ differentially private algorithms above results in $(\epsilon_g,\delta_g)$-differentially private for a multitude of possible $(\epsilon_g,\delta_g)$ pairs. 
Murtagh and Vadhan showed that determining the minimal $\epsilon_g$ given $\delta_g$ is $\#\P$-complete~\cite{murtagh2016complexity}. They also gave a polynomial time approximation algorithm that computes $\epsilon_g$ to arbitrary accuracy, giving hope that for ``simple'' programs deciding differential privacy or approximating of privacy parameters may be tractable. Unfortunately, our results show that this is not the case.

\subsection{Contributions}

Following the literature, we refer to the variant of differential privacy where $\delta=0$ as {\em pure} differential privacy and to the variant where $\delta>0$ as {\em approximate} differential privacy. We contribute in three directions. 

\begin{itemize}
\item \textbf{Bounding pure differential privacy.}  We show that determining whether a randomized circuit is $\epsilon$-differentially private is $\conpsp$-complete.\footnote{The class $\conpsp$ is contained in $\PSPACE$ and contains the polynomial hierarchy (as, per Toda's Theorem, $\PH\subseteq \psp$).\vspace{-0.8cm}} To show hardness in $\conpsp$ we consider a complement to the problem $\ems$~\cite{littman1998computational}, which is complete for $\NP^{\#\P}$~\cite{chistikov2017approximate}. In the complementary problem, $\ams$, given a formula $\generalformula{}$ over $n+m$ variables the task is to determine if for all allocations $\xa \in\zo{n}$, $\phi(\xa,\ya)$ evaluates to true on no more than $\frac{1}{2}$ of allocations to $\ya \in {\zo{m}}$.

\item \textbf{Bounding approximate differential privacy.} Turning to the case where $\delta >0$, we show that determining whether a randomized circuit is $(\epsilon,\delta)$-differentially private is $\conpsp$-complete when the number of output bits is small relative to the total size of the circuit and otherwise between $\conpsp$ and $\conpspsp{}$. 

\item \textbf{Approximating the parameters $\epsilon$ and $\delta$.} Efficient approximation algorithms exist for optimal composition~\cite{murtagh2016complexity}, and one might expect the existence of polynomial time algorithms to approximate $\epsilon$ or $\delta$ in randomized circuits. We show this is $\NP$-hard and $\coNP$-hard, and therefore an efficient algorithm does not exist (unless $\P = \NP$).

\end{itemize}Our results show that for loop-free programs with probabilistic choice directly verifying whether a program is differentially private is intractable. These results apply to programs in any reasonable programming language supporting randomized computations. Hence, they set the limits on where to search for automated techniques for these tasks.   

\paragraph*{The relation to quantitative information flow} Differential privacy shares similarities with quantitative information flow~\cite{Denning82,Gray90}, which is an entropy-based theory measuring how secure a program is. Alvim et al.~\cite{alvim2011relation} showed that a bound on pure differential privacy implies a bound on quantitative information flow. So, one could hope that bounding differential privacy could be easier than bounding quantitative information flow. 
Yasuoka and Terauchi~\cite{yasuoka2010bounding} have shown that bounding quantitative information flow for loop free boolean programs with probabilistic choice is $\PP$-hard (but in $\PSPACE$). In contrast, our results show that 
bounding pure differential privacy is $\conpsp$-complete. 
Chadha et al.~\cite{chadha2014quantitative} showed the problem  to be $\PSPACE$-complete for boolean programs with loops 
and probabilistic choice (notice that this would be not true for programs with integers). We leave the analogous question for future works.

\section{Preliminaries}
\label{sec:prelims}
\paragraph*{Numbers} By a \textit{number given as a rational} we mean a number of the form $\frac{x}{y}$ where $x,y$ are given as binary integers. %

\subsection{Loop-free probabilistic programs}\label{sec:lfpp}
We consider a simple loop-free imperative programming language built over Booleans, and including probabilistic choice. 
\[\begin{array}{rcl@{\qquad}r}
x&::=& [\text{a}{-}\text{z}]^+ & \text{(variable identifiers)} \\
b&::=& {\tt true}\ |\ {\tt false}\ |\ {\tt random}\ |\ x\ |\  b \wedge b\ |\ b \vee b\ |\  \neg b\  & \text{(boolean expressions)} \\
c&::=& {\tt SKIP}\ |\ x := b\ |\ c;\ c\ |\ {\tt if}\ b\ {\tt then}\ c\ {\tt else}\ c & \text{(commands)} \\
t&::=& x\ |\ t, x & \text{(list of variables)} \\
p&::=& {\tt input}(t);\  c;\ {\tt return}(t) & \text{(programs)} 
\end{array}
\] \newpage

Probabilistic programs \cite{DBLP:journals/jcss/Kozen81} extend standard programs with the addition of coin tosses; this is achieved by the probabilistic operation $\texttt{random}$, which returns either \texttt{true} or \texttt{false} with equal probability. A standard operation, sometimes denoted by $c \oplus c$, which computes one of the two expressions with probability $\frac{1}{2}$ each is achieved with ${\tt if\ random\ then }\ c\ {\tt else}\ c$. The notation $c \oplus c$ is avoided as $\oplus$ refers to \textit{exclusive or} in this paper.

The semantics of the programming language are standard and straight forward. Without loss of generality, each variable assignment is final, that is, each assignment must go to a fresh variable. Looping behaviour is not permitted, although bounded looping can be encoded by unrolling the loop. 

\begin{remark} \label{remark:integers} Our results also hold when the language additionally supports integers and the associated operations (e.g. ${+},{\times},{-},{\ge},{=}$ etc.), providing the integers are of a bounded size. 
Such a language is equally expressive as the language presented here. Further details are given in the full version of the paper.
\end{remark}

\subsection{Probabilistic circuits}
\begin{figure*}[t]
\centering
\includegraphics[width=0.8\linewidth]{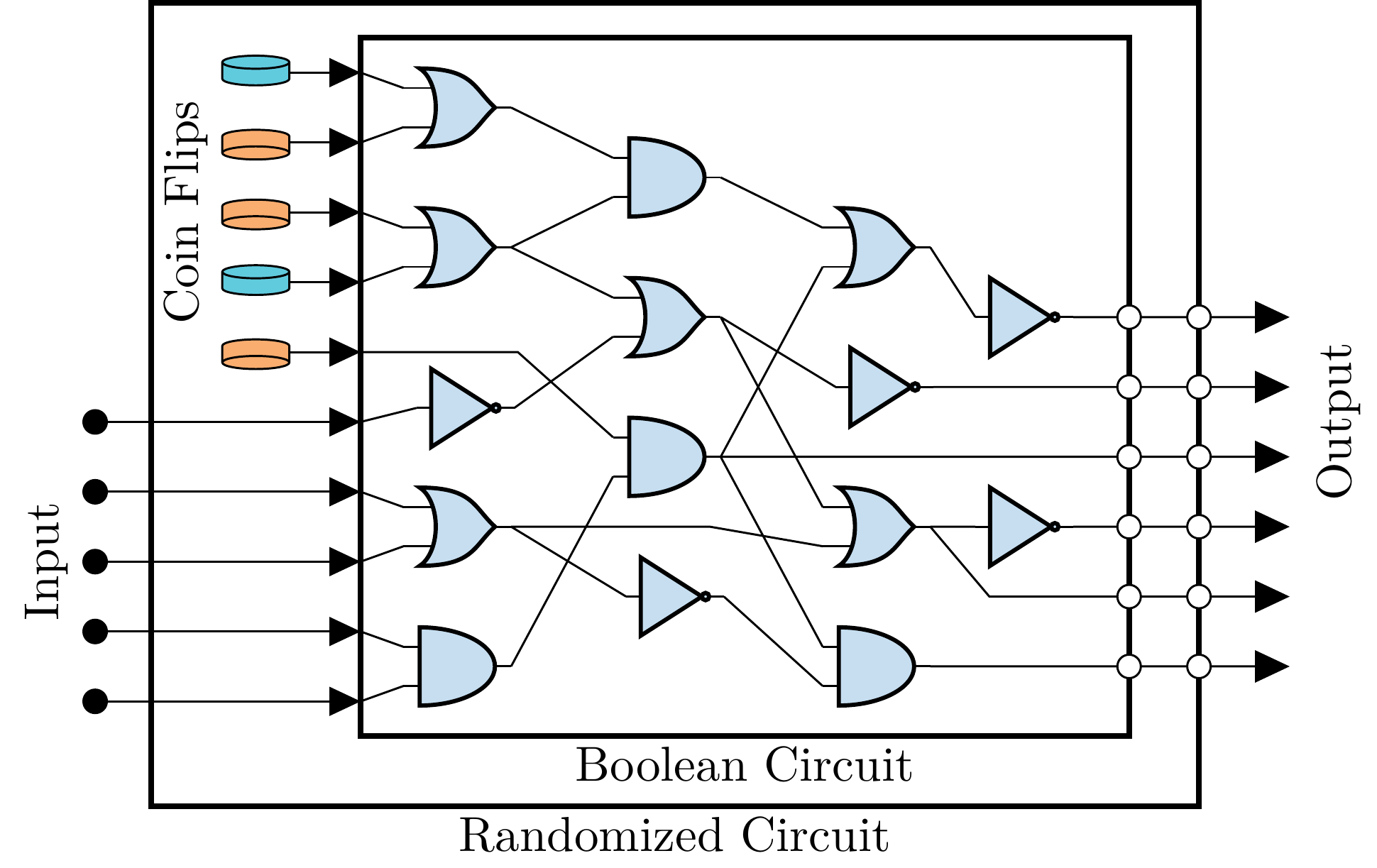}
\caption{Example randomized circuit.}
\label{fig:egrc}
\end{figure*}
\begin{definition}\label{defn:probcircuit}
A Boolean circuit $\generalcircuit{}$ with $n$ inputs and $\ell$ outputs is a directed acyclic graph $\generalcircuit{}=(V,E)$ containing $n$ input vertices with zero in-degree, labeled $X_1,\ldots,X_n$ and $\ell$ output vertices with zero out-degree, labeled $O_1,\ldots,O_\ell$. Other nodes are assigned a label in $\{\wedge, \vee, \neg\}$, with vertices labeled $\neg$ having in-degree one and all others having in-degree two. The size of $\generalcircuit{}$, denoted $|\generalcircuit{}|$, is defined to be $|V|$. 
A {\em randomized circuit} has $m$ additional random input vertices labeled $R_1,\ldots,R_m$.

Given an input string $\xa{}=(x_1,\ldots,x_n)\in\zo{n}$, the circuit is evaluated as follows. First, the values $x_1,\ldots,x_n$ are assigned to the nodes labeled $X_1,\ldots,X_n$. Then, $m$ bits $\proballoc{}=(r_1,\ldots,r_m)$ are sampled uniformly at random from $\zo{m}$ and assigned to the nodes labeled $R_1,\ldots,R_m$. Then, the circuit is evaluated in topological order in the natural way. E.g., let $v$ be a node labeled $\wedge$ with incoming edges $(u_1,v), (u_2,v)$ where $u_1, u_2$ were assigned values $z_1,z_2$ then $v$ is assigned the value $z_1\wedge z_2$. The outcome of $\generalcircuit{}$ is $(o_1,\ldots,o_\ell)$,  the concatenation of values assigned to the $\ell$ output vertices $O_1,\ldots,O_\ell$.
\goodbreak
For input $\xa{}\in\zo{n}$ and event $E \subseteq \zo{\ell}$ we have
\lipicscenter{\pr[\generalcircuit{}(\xa{}) \in E] = \frac{|\{\proballoc{}\in\zo{m} \; : \; \generalcircuit{}(\xa{}, \proballoc{}) \in E \}|}{2^m}.} 
\end{definition}

\begin{remark}
The operators, $\wedge, \vee$ and $\neg$ are functionally complete. However, we will also use $\oplus$ (exclusive or), such that $p \oplus q \iff  (p \vee q) \wedge  \neg (p \wedge q)$. \end{remark}

\subsection{Equivalence of programs and circuits}

\begin{lemma}\label{lemma:circuiteqprogram}
A \textit{loop-free probabilistic program} can be converted into an equivalent \textit{probabilistic boolean circuit} in linear time in the size of the program (and vice-versa).
\end{lemma}
\begin{proof}[Proof sketch]
It is  clear that  a probabilistic circuit can be expressed as a probabilistic program  using just  boolean operations by expressing a variable for each vertex after sorting the vertices in topological order.

To convert a probabilistic Boolean program into a probabilistic circuit, each of the commands can be handled using a fixed size sub-circuit, each of which can be composed together appropriately.
\end{proof}

Given the equivalence between loop-free probabilistic programs and probabilistic circuits, the remainder of the paper will use probabilistic circuits.

\subsection{Differential privacy in probabilistic circuits}

Let $X$ be any input domain. An input to a differentially private analysis would generally be an array of elements from a data domain $X$, each corresponding to the information of an individual, i.e., $\xa=(x_1,\ldots,x_n)\in X^n$. 

The definition of differential privacy depends on adjacency between inputs, we define \textit{neighboring} inputs.

\begin{definition}
Inputs $\xa=(x_1,\ldots,x_n)$ and $\xa'=(x'_1,\ldots,x'_n)\in X^n$ are called {\em neighboring} if there exist $i\in[n]$ s.t.\ if $j \ne i $ then $x_j = x'_j$.
\end{definition}

In this work, we will consider input domains with finite representation. Without loss of generality we set $X=\zo{k}$ and hence an array $x=(x_1,\ldots,x_n)$ can be written as a sequence of $nk$ bits, and given as input to a (randomized) circuit with $nk$ inputs. Our lower bounds work already for for $k = 1$ and our upper bounds are presented using $k = 1$ but generalise to all $k$.

\begin{definition}[Differential Privacy~\cite{dwork2006calibrating,dwork2006our}]
A probabilistic circuit $\generalcircuit{}$ is \eddpadj{} if for all neighboring $\xa{},\xa{}'\in X^n$ and for all $E \subseteq\zo{\ell}$,
\lipicscenter{\pr[\psi(\xa{}) \in E] \le e^\epsilon \cdot \pr[\psi(\xa{}') \in E] + \delta.}
\end{definition}

Following common use, we refer to the case where $\delta=0$ as {\em pure} differential privacy and to the case where $\delta >0$ as {\em approximate} differential privacy. When omitted, $\delta$ is understood to be zero.

\goodbreak
\subsection{Problems of deciding and approximating differential privacy}

We formally define our three problems of interest.
\begin{definition}
The problem \decideedp asks, given $\epsilon$ and $\generalcircuit{}$, if $\generalcircuit{}$ is $\epsilon$-differentially private. We assume $\epsilon$ is given by the input $e^\epsilon$ as a rational number.
\end{definition}

\begin{definition}
The problem \decideeddp asks, given $\epsilon$, $\delta$ and $\generalcircuit{}$, if $\generalcircuit{}$ is \eddpadj{}. We assume $\epsilon$ is given by the input $e^\epsilon$ as a rational number.
\end{definition}

\begin{definition} \label{def:approx}
Given an approximation error %
$\gamma >0$, the \textsc{Approximate}-$\delta$ problem and the \textsc{Approximate}-$\epsilon$ problem, respectively, ask: 
\begin{itemize}
  \item Given $\epsilon$, find $\hat \delta \in [0,1]$, such that $0\le \hat\delta-\delta \le \gamma$, where $\delta$ is the minimal value such that $\psi$ is \eddpadj{}.
  \item Given $\delta$, find $\hat \epsilon \geq 0$, such that $0\le \hat \epsilon -\epsilon \le \gamma$,  where $\epsilon$ is the minimal value such that $\psi$ is \eddpadj{}.
\end{itemize}
\end{definition}

\subsection{The class \conpsp}

The complexity class $\#\P$ is the counting analogue of $\NP$ problems. In particular $\ssat$, the problem of counting the number of satisfying assignments of a given a boolean formula $\generalformula{}$ on $n$ variables, is complete for $\#\P$. Similarly $\ccsat{}$, the problem of counting the satisfying assignments of a circuit with a single output, is complete for $\#\P$.

A language $L$ is in \conpsp{} if membership in $L$ can be refuted using a polynomial time non-deterministic Turing machine with access to a $\#\P$ oracle.  %
It is easy to see that $\conpsp{} = \coNP^{\PP}$, and 
$\PH\subseteq \conpsp{} \subseteq \PSPACE$, where $\PH\subseteq \conpsp{}$ follows by Toda's theorem ($\PH\subseteq \P^{\#\P}$)~\cite{toda1991pp}.

The following decision problem is complete for $\NP^{\#\P}$~\cite{chistikov2017approximate}:
\begin{definition}
\ems asks, given $\generalformula{}$ a quantifier free formula over $n+m$ variables if there exist an allocation $\xa \in\zo{n}$ such that there are strictly greater than $\frac{1}{2}$ of allocations to $\ya \in {\zo{m}}$ where $\phi(\xa,\ya)$ evaluates to true.
\end{definition}

The complementary problem $\ams$, is complete for $\conpsp{}$: a formula $\generalformula{}$ is $\ams$, if $\generalformula{}$ is not $\ems$. That is, $\generalformula{}$ a quantifier free formula over $n+m$ variables is $\ams$ if for all allocations $\xa \in\zo{n}$ there are no more than $\frac{1}{2}$ of allocations to $\ya \in {\zo{m}}$ where $\phi(\xa,\ya)$ evaluates to true.

\section{The complexity of deciding pure differential privacy}
\label{sec:edp}

In this section we classify the complexity of deciding $\epsilon$-differential privacy, for which we show the following theorem:

\begin{theorem}\label{thm:dedpconpsp}
\decideedp is $\conpsp$-complete.
\end{theorem}

It will be convenient to consider the well-known simpler reformulation of the definition of pure differential privacy in finite ranges to consider specific outcomes $\outpt\in \zo{\ell}$ rather than events $E\subseteq\zo{\ell}$.

\begin{reformulation}[Pure differential privacy]
A probabilistic circuit $\generalcircuit{}$ is $\epsilon$-differentially private if and only if for all neighboring $\inp{},\inpp{}\in X^n$  and for all $\outpt{} \in \zo{\ell}$,
\lipicscenter{\pr[\psi(\inp{}) = \outpt{}] \le e^\epsilon \cdot \pr[\psi(\inpp{}) = \outpt{}]. }
\end{reformulation}

\subsection{\decideedp is in \conpsp}

We show a non-deterministic Turing machine which can `refute' $\generalcircuit{}$ being $\epsilon$-differentially private in (non-deterministic) polynomial time with a $\#\P$ oracle. 
A circuit $\generalcircuit{}$ is shown not to be  $\epsilon$-differentially private by exhibiting a combination $\inp,\inpp,\outpt{}$ such that 
$\pr[\psi(\inp) =\outpt{}] > e^\epsilon \cdot \pr[\psi(\inpp) = \outpt{}]. $ The witness to the non-deterministic Turing machine would be a sequence of $2n$ bits parsed as neighboring inputs $\inp{},\inpp{} \in \zo{n}$ and $\ell$ bits describing an output $\outpt{} \in \zo{\ell}$. 
The constraint can then be checked in polynomial time, using the $\#\P$ oracle to compute $\pr[\psi(\inp) =\outpt{}]$ and $\pr[\psi(\inpp) =\outpt{}]$. 

To compute $\pr[\psi(\inp) =\outpt{}]$ in $\#\P$ we create an instance to $\ccsat$, which will count the number of allocations to the $m$ probabilistic bits consistent with this output. We do this by extending $\generalcircuit{}$ with additional gates reducing to a single output which is true only when the input is fixed to $\inp$ and the output of $\generalcircuit{}$ was $\outpt{}$. 

\subsection{\conpsp-hardness of \decideedp}

To show  $\conpsp$-hardness of $\decideedp{}$ we show a reduction from  $\ams$ in \cref{lem:amstodecideedp}; together with the inclusion result above, this entails that $\decideedp{}$ is $\conpsp$-complete (\cref{thm:dedpconpsp}). 

Randomized response~\cite{warner1965randomized} is a technique for answering sensitive Yes/No questions by flipping the answer with probability $p \le 0.5.$ Setting $p= \frac{1}{1 + e^\epsilon}$ gives $\epsilon$-differential privacy. Thus $p=0$ gives no privacy and $p=0.5$ gives total privacy (albeit no utility).

\begin{definition}[Randomized Response]\label{section:aboutrr}
\[RR_\epsilon(x) = \begin{cases}  x & \text{w.p.}\quad \frac{e^\epsilon}{1 + e^\epsilon} \\ \neg x & \text{w.p.}\quad  \frac{1}{1 + e^\epsilon}\end{cases}\]
\end{definition}

\begin{lemma}\label{lem:amstodecideedp}\ams reduces in polynomial time to \decideedp.
\end{lemma}
\begin{proof}
We will reduce from \ams{} to \decideedp{} using randomized response. We will take a boolean formula $\generalformula{}$ and create a probabilistic circuit that is $\epsilon$-differentially private if and only if $\generalformula{}$ is \ams.

Consider the circuit $\generalcircuit{}$ which takes as input the value $z \in \zo{}$. It probabilistically chooses a value of $\xa{}\in \zo{n}$ and $\ya{}\in\zo{m}$ and one further random bit $p_1$ and computes 
$b=z\oplus \neg(p_1 \vee \phi(\xa{},\ya{}))$. The circuit outputs $(\xa{}, b)$.

\begin{claim}$\generalcircuit{}$ is $\ln(3)$-differentially private if and only if $\generalformula{}$ is \ams.
\end{claim}

Suppose $\phi \in \ams$ then, no matter the choice of $\xa{}$, 
\lipicscenter{0 \leq \Pr_y[\phi(\xa{},\ya{})=1] \leq \frac{1}{2},} and hence 
\lipicscenter{ \frac{1}{4}\leq \Pr_{y,p_1}[\neg(p_1 \vee \phi(\xa{},\ya{}))=1] \leq \frac{1}{2}.}  
We conclude the true answer $z$ is flipped between $\frac{1}{4}$ and $\frac{1}{2}$ of the time, observe this is exactly the region in which randomized response gives us the most privacy.  In the worst case  $p =\frac{1}{4} = \frac{1}{1+e^\epsilon}$, gives $e^\epsilon = 3$, so $\ln(3)$-differential privacy.
\goodbreak

In the converse, suppose $\phi \in \ems$, then for some  $\xa{}$ 
$$ \frac{1}{2} < \Pr_y[\phi(\xa{},\ya{})=1] \leq 1,$$ and 
then $$\Pr_{y,p_1}[\neg(p_1 \vee \phi(\xa{},\ya{}))=1] < \frac{1}{4},$$ in which case the randomized response does not provide $\ln(3)$-differential privacy.
\end{proof}

\begin{remark}
We skew the result so that in the positive case (when  $\phi \in \ams$) the proportion of accepting allocations is between $\frac{1}{4}$ and $\frac{1}{2}$, resulting in the choice of $\ln(3)$-differentially privacy. Alternative skews, using more bits akin to $p_1$, shows hardness for other choices of $\epsilon$.
\end{remark}

\subsubsection{Hardness by circuit shape}
In our proof of the upper-bound we use $\coNP$ to resolve the non-deterministic choice of both input and output. We show this is necessary in the sense $\coNP$ is still required for either large input or large output. 
The hardness proof used in \cref{lem:amstodecideedp} shows that when $|\psi| = n$ the problem is hard for $\Omega(1)$-bit input and $\Omega(n)$-bit output.

We can also prove this is hard for $\Omega(n)$-bit input and $\Omega(1)$-bit output. Intuitively a counter example to differential privacy has two choices: a pair of adjacent input and a given output upon which the relevant inequality will hold. So to ``refute'' $\ams$ the counterexample of the $\textsc{All}$ choice (i.e. $\xa$) can be selected in the input, rather than the output as in our case. Since the input is now non-trivial we must take care of what happens when the adjacent bit is in the choice of $\xa$. Details are given in the full version.

Further the problem is in $\P^{\#\P}$ for $O(\log(n))$-bit input and $O(\log(n))$-bit output, as in this case, the choices made by $\coNP$ can instead be checked deterministically in polynomial time. In this case we show $\PP$-hardness, which applies even when there is 1-bit input and 1-bit output.

\section{On the complexity of deciding approximate differential privacy}
\label{sec:eddp}

It is less clear whether deciding $(\epsilon,\delta)$-differential privacy can be done in $\conpsp{}$. First we consider restrictions to the shape of the circuit so that $\conpsp{}$ can be recovered, and then show that in general the problem is in $\conpspsp{}$. 

Recall that in the case of  $\epsilon$-differential privacy it was enough to consider singleton events $\{\outpt{}\}$ where $\outpt{}\in \zo{\ell}$, however in the definition of \eddpnoun{} we must quantify over output events $E\subseteq \zo{\ell}$. If we consider circuits with one output bit ($\ell = 1$), then the event space essentially reduces to $E \in \{\emptyset,\{0\},\{1\},\{0,1\}\}$ and we can apply the same technique. 

We  expand this to the case when the number of outputs bits is logarithmic $\ell \le \log(|\generalcircuit{}|)$. To cater to this, rather than guessing a violating $E\in \zo{\ell}$, we consider a violating subset of events $E\subseteq\zo{\ell}$. Given such an event $E$ we create a circuit $\generalcircuit{}_E$ on $\ell$ inputs and a single output which indicates whether the input is in the event $E$. The size of this circuit is exponential in $\ell$, thus polynomial in $|\generalcircuit{}|$. Composing $\generalcircuit{}_E \circ \generalcircuit{}$, we check the conditions hold for this event $E$, with just one bit of output.

\begin{claim} \label{clm:dedideddp_small_output_space}
\decideeddp{}, restricted to circuits $\generalcircuit$ with $\ell$ bit outputs where $\ell \le \log(|\generalcircuit|)$,  is in \conpsp{} (and hence \conpsp{}-complete).
\end{claim}

The claim trivially extends to $\ell \le c\cdot\log(|\generalcircuit|)$ for any {\em fixed} $c>0$.  

\goodbreak

\subsection{\decideeddp{} is in \conpspsp{}}

We now show that $\decideeddp{}$ in the most general case can be solved in $\conpspsp{}$. We will assume $e^\epsilon = \alpha$ is given as a rational, with $\alpha = \frac{u}{v}$ for some integers $u$ and $v$. Recall we use $n,\ell$ and $m$ to refer to the number of input, output and random bits of a circuit respectively.
While we will use non-determinism to choose inputs leading to a violating event, unlike in Section~\ref{sec:edp} it would not be used for finding a violating event $E$, as an (explicit) description of such an event may be of super-polynomial length. It would be useful for us to use a reformulation of approximate differential privacy, using a sum over potential individual outcomes.

\begin{reformulation}[Pointwise differential privacy~\cite{barthe2016proving}\label{lem:pointwiseeddp}]
A probabilistic circuit $\generalcircuit{}$ is $(\epsilon,\delta)$-differentially private if and only if for all neighboring $\xa,\xa'\in X^n$  and for all $\outpt\in \zo{\ell}$,
\lipicscenter{\sum_{ \outpt\in{\zo{\ell}}} \delta_{\xa,\xa'}(\outpt)  \le \delta,}%
where $\ \delta_{\xa,\xa'}(\outpt) = 
\max\left(\pr[\psi(\xa) = \outpt] - e^\epsilon \cdot \pr[\psi(\xa') = \outpt], 0\right). $
\end{reformulation}

We define $\mathcal{M}$, a non-deterministic Turing Machine with access to a $\#\P$-oracle, and where each execution branch runs in polynomial time.  On inputs a probabilistic circuit $\generalcircuit$ and neighboring $\xa,\xa'\in X^n$ the number of accepting executions of  $\mathcal{M}$ would be proportional to $\sum_{ \outpt\in{\zo{\ell}}} \delta_{\xa,\xa'}(\outpt)$.

In more detail, on inputs $\generalcircuit$, $\xa$ and $\xa'$, $\mathcal{M}$ chooses  $\outpt \in \zo{\ell}$ and an integer $C \in \{1,2,\dots, 2^{m+\bits{v}}\}$ (this requires choosing $\ell+m+\bits{v}$ bits). 
Through a call to the $\#\P$ oracle, $\mathcal{M}$ computes 
\lipicscenter{ a  =  \left|\{\proballoc{}\in\zo{m} : \psi(\xa,\proballoc{}) = \outpt\}\right|} and \lipicscenter{b  =  \left|\{\proballoc{}\in\zo{m} : \psi(\xa',\proballoc{}) = \outpt\}\right|. }
Finally, $\mathcal{M}$ accepts if  $v \cdot a - u\cdot b \ge C$ and otherwise rejects.

\begin{lemma} \label{lemma:oracleacceptingpaths}
Given two inputs $\xa,\xa'\in X^n$, $\mathcal{M}(\psi,\xa,\xa'$) has exactly $v\cdot 2^m  \sum_{ \outpt\in{\zo{\ell}}} \delta_{\xa,\xa'}(\outpt)$ accepting executions.
\end{lemma}

{\allowdisplaybreaks
\begin{proof}
Let $\ind{X}$ be the indicator function, which is one if the predicate $X$ holds and zero otherwise.
\begin{align*}
&v \cdot 2^m  \sum_{\outpt \in \zo{\ell}}  \delta_{\xa,\xa'}(\outpt{})= \sum_{\outpt \in \zo{\ell}}v\cdot 2^m \max\left( \pr[\psi(\xa) = \outpt] - \alpha \pr[\psi(\xa') = \outpt], 0 \right) 
\\&= \sum_{\outpt \in \zo{\ell}}v2^m \max\left( \frac{1}{2^m}\sum_{\proballoc\in\zo{m}}\ind{\psi(\xa,\proballoc) = \outpt} - \alpha \frac{1}{2^m}\sum_{\proballoc\in\zo{m}}\ind{\psi(\xa',\proballoc) = \outpt}, 0 \right) 
\\
&= \sum_{\outpt \in \zo{\ell}} \max\left( v\sum_{\proballoc\in\zo{m}}\ind{\psi(\xa,\proballoc) = \outpt} - v\alpha \sum_{\proballoc\in\zo{m}}\ind{\psi(\xa',\proballoc) = \outpt}, 0 \right) 
\\&= \sum_{\outpt \in \zo{\ell}} \max\left( v\sum_{\proballoc\in\zo{m}}\ind{\psi(\xa,\proballoc) = \outpt} - u \sum_{\proballoc\in\zo{m}}\ind{\psi(\xa',\proballoc) = \outpt}, 0 \right) 
\\&= \sum_{\outpt \in \zo{\ell}} \mkern-12mu \sum_{C= 1}^{2^{\bits{v}+m}} \ind{\max\left( v\sum_{\proballoc\in\zo{m}}\mkern-12mu\ind{\psi(\xa,\proballoc) = \outpt} - u \sum_{\proballoc\in\zo{m}}\mkern-12mu\ind{\psi(\xa',\proballoc) = \outpt}, 0 \right) \ge C}
\\&=\text{number of accepting executions in } \widehat{\mathcal{M}}\qedhere
\end{align*}\end{proof}}

We can now describe our $\conpspsp{}$ procedure for $\decideeddp{}$. The procedure takes as input a probabilistic circuit $\generalcircuit$.
\begin{enumerate}
\item Non-deterministically choose neighboring $\xa$ and $\xa' \in \zo{n}$ (i.e.,  $2 n$ bits).
\item Let $\mathcal{M}$ be the non-deterministic Turing Machine with access to a $\#\P$-oracle as described above. Create a machine $\widehat{\mathcal{M}}$ with no input that executes $\mathcal{M}$ on $\generalcircuit, \xa,\xa'$.
\item Make an $\#\P^{\#\P}$ oracle call  for the number of accepting executions $\widehat{\mathcal{M}}$ has. 
\item Reject if the number of accepting executions is greater than $v \cdot  2^m \cdot \delta$ and otherwise accept.
\end{enumerate}

By Lemma~\ref{lemma:oracleacceptingpaths}, there is a choice $\xa,\xa'$ on which the procedure rejects if and only if $\generalcircuit$ is not $(\epsilon,\delta)$-differentially private.

\subsection{Hardness}%
\cref{thm:dedpconpsp} shows that $\decideedp{}$ is $\conpsp{}$-complete, in particular $\conpsp{}$-hard and since $\decideedp{}$ is a special case of $\decideeddp{}$, this is also $\conpsp{}$-hard. Nevertheless the proof is based on particular values of $\epsilon$ and in the full version  we provide an alternative proof of hardness based on $\delta$. This proof result will apply for any $\epsilon$ (even for $\epsilon = 0$) and for a large range of $\delta$ (but not $\delta = 0$).

The proof proceeds by first considering the generalisation of $\ams$ to the version where \textit{minority}, i.e. less than $\frac{1}{2}$ of the assignments, is replaced with another threshold. This problem is also $\conpsp{}$-hard for a range of thresholds. Note however, if this threshold is exactly $1$ the problem is true for all formulae, and if the threshold is $0$ the problem is simply asks if the formula is unsatisfiable (a $\coNP$ problem).

This generalised problem can then be reduced to deciding $\decideeddp{}$, where the threshold corresponds exactly to $\delta$. It will turn out in the resulting circuit $\epsilon$ does not change the status of differential privacy, i.e. it is $(\epsilon,\delta)$-differentially private for all $\epsilon$, or not.

The proof shows hardness for $\Omega(n)$-input bits and $1$-output bit; the case in which there also exists a $\conpsp{}$ upper-bound. Hence, showing hardness in a higher complexity class, e.g., $\conpspsp{}$,  would require a reduction to a circuit with more output bits.

\section{Inapproximability of the privacy parameters \texorpdfstring{$\epsilon, \delta$}{epsilon, delta}}
\label{sec:approx}

Given the difficulty of deciding if a circuit is differentially private, one might naturally consider whether approximating $\epsilon$ or $\delta$ could be efficient. We show that these tasks are both $\NP$-hard and $\coNP$-hard. 

We show that distinguishing between $(\epsilon,\delta)$, and \eddpnoun{'} is $\NP$-hard, by reduction from a problem we call $\textsc{Not-Constant}$ which we also show is $\NP$-hard. A boolean formula is in $\textsc{Not-Constant}$ if it is satisfiable and not also a tautology. %

\begin{lemma}\label{lem:satnottautnphard}
\textsc{Not-Constant} is $\NP$-complete.  (hence \textsc{Constant} is $\coNP$-complete).
\end{lemma}
\goodbreak
\begin{proof}
Clearly, $\textsc{Not-Constant}\in\NP$, the witness being a pair of satisfying and non-satisfying assignments. We reduce 3-SAT to \textsc{Not-Constant}. Given a Boolean formula $\phi$ over variables $x_1,\ldots,x_n$ let $\phi'(x_1,\ldots,x_n,x_{n+1})=\phi(x_1,\ldots, x_n) \wedge x_{n+1}$. Note that $\phi'$ is never a tautology as $\phi'(x_1,\ldots,x_n,0) = 0$. Furthermore, $\phi'$ is satisfiable iff $\phi$ is.
\end{proof}

In \cref{section:aboutrr} we used randomized response in the pure differential privacy setting. 
We now consider the approximate differential privacy variant $RR_{\epsilon,\delta}: \{0,1\}\rightarrow\{\top,\bot\}\times\{0,1\}$ defined as follows:
$$RR_{\epsilon,\delta}(x) = \begin{cases}
(\top,x) & \mbox{w.p.}~\delta \\
(\bot, x) & \mbox{w.p.}~(1-\delta) \frac{\alpha}{1+\alpha} \\
(\bot,\neg x) & \mbox{w.p.}~(1-\delta) \frac{1}{1+\alpha}\end{cases}~~~ \text{ where }~~~ \alpha = e^\epsilon
$$

I.e., with probability $\delta$, $RR_{\epsilon,\delta}(x)$ reveals $x$ and otherwise it executes $RR_\epsilon(x)$. The former is marked with ``$\top$'' and the latter with ``$\bot$''. This mechanism is equivalent to the one described in~\cite{murtagh2016complexity} and is \eddpadj{}.

\begin{definition}
Let $0\le \epsilon \le \epsilon'$, $0\le \delta \le \delta'\le 1$, with either $\epsilon < \epsilon'$ or $\delta < \delta'$. The problem \disinguished takes as input a circuit $\generalcircuit{}$, guaranteed to be either \eddpadj{}, or \eddpadj{'}. The problem asks whether $\generalcircuit{}$ is \eddpadj{} or $(\epsilon',\delta')$-differentially private.
\end{definition}

\begin{lemma}
\disinguished is $\NP$-hard (and $\coNP$-hard).
\end{lemma}
\begin{proof}
We reduce \textsc{Not-Constant} to \disinguished.
Given the boolean formula $\phi(\xa)$ on $n$ bits, we create a probabilistic circuit $\psi$.  The input to $\psi$ consists of the $n$ bits $\xa$ plus a single bit $y$. The circuit $\psi$ has four output bits $(o_1,o_2,o_3,o_4)$ such that $(o_1,o_2) = RR_{\epsilon,\delta}(y)$ and $(o_3,o_4) = RR_{\epsilon',\delta'}(\phi(\xa))$.

Observe that $(o_1,o_2) = RR_{\epsilon,\delta}(y)$ is always $(\epsilon,\delta)$ differentially private. As for $(o_3,o_4) = RR_{\epsilon',\delta'}(\phi(\xa))$, 
if $\phi \in \textsc{Not-Constant}$ then there are adjacent $\xa,\xa{}'$ such that $\phi(\xa) \ne \phi(\xa{}')$. In this case, $(o_3,o_4) = RR_{\epsilon',\delta'}(\phi(\xa))$ is $(\epsilon',\delta')$-differentially private, and, because $(\epsilon,\delta) < (\epsilon',\delta')$, so is $\psi$ . On the other hand, if  $\phi \not\in\textsc{Not-Constant}$ then $\phi(\xa)$ does not depend on $\xa$ and hence $(o_3,o_4)$ does not affect privacy, in which case we get that $\psi$ is $(\epsilon,\delta)$ differentially private. 

The same argument also gives $\coNP$-hardness.
\end{proof}

Notice that the above theorem holds when $\delta = \delta'$ and $\epsilon < \epsilon'$ (similarly, $ \epsilon=\epsilon'$ and $\delta < \delta'$), which entails the following theorem:
\begin{theorem}\label{lem:approxhard}
Assuming $\P\ne \NP$, for any approximation error $\gamma > 0$, there does not exist a polynomial time approximation algorithm that given a probabilistic circuit $\psi$ and $\delta$   computes some $\hat{\epsilon}$, where $|\hat{\epsilon} - \epsilon| \le \gamma$  and $\epsilon$  is the minimal such that $\psi$ is $(\epsilon,\delta)$-differentially private within error $\gamma$. Similarly, given $\epsilon$, no such  $\hat{\delta}$ can be computed polynomial time where $|\hat{\delta} - \delta| \le \gamma$ and  $\delta$ is minimal.
\end{theorem}
\begin{remark} 
The result also applies when approximating within a given ratio $\rho> 1$ (e.g. in the case of approximating $\epsilon$, to find $\hat{\epsilon}$ such that $\frac{\hat{\epsilon}}{\epsilon} \le \rho$). Moreover, the result also holds when approximating pure differential privacy, that is when $\delta = 0$.
\end{remark}

\goodbreak

\section{Related work}

Differential privacy was introduced in~\cite{dwork2006calibrating}. It is a definition of privacy in the context of data analysis capturing the intuition that information specific to an individuals is protected if every single user's input has a bounded influence on the computation's outcome distribution, where the bound is specified by two parameters, usually denoted by $\epsilon,\delta$. Intuitively, these parameters set an upperbound on privacy loss, where the parameter $\epsilon$ limits the loss and the parameter $\delta$ limits the probability in which the loss may exceed $\epsilon$.

Extensive work has occurred in the computer-assisted or automated of verification of differential privacy. Early work includes, PINQ~\cite{mcsherry2009privacy} and Airavat~\cite{roy2010airavat} which are systems that keep track of the privacy budgets ($\epsilon$ and $\delta$) using trusted privacy primitives in SQL-like and MapReduce-like paradigms respectively. In other work, programming languages were developed, that use the type system to keep track of the sensitivity and ensure the correct level of noise is added~\cite{reed2010distance,barthe2012probabilistic,d2013sensitivity,barthe2016programming}. Another line of work uses proof assistants to help prove that an algorithm is differentially private~\cite{barthe2016proving};  although much of this work is not automated, recent work has gone in this direction~\cite{albarghouthi2017synthesizing,zhang2017lightdp}. 

These techniques focuses on `soundness', rather than `completeness' thus are not amenable to complexity analysis. In the constrained case of verifying differential privacy on probabilistic automata and Markov chains there are bisimulation based techniques~\cite{tschantz2011formal,chatzikokolakis2014generalized}. Towards complexity analysis;~\cite{chistikov2019asymmetric} shows that computing the optimal value of $\delta$ for a finite labelled Markov chain is undecidable. Further~\cite{chistikov2018bisimilarity} and \cite{chistikov2019asymmetric} provides distances, which are (necessarily) not tight, but can be computed in polynomial time with an $\NP$ oracle and a weaker bound in polynomial time.
Recent works have focused on developing techniques for finding violations of differential privacy~\cite{DingWWZK18,BichselGDTV18}. The methods proposed so far have been based on some form of testing. Our result limits also the tractability of these approaches. Finally, \cite{barthe2019automated}  proposes an automated technique for proving differential privacy or finding counterexamples. This paper studies a constrained class of programs extending the language we presented here, and provides a `complete' procedure for deciding differential privacy for them. The paper does not provide any complexity guarantee for the proposed method and we expect our results to apply also in their setting.

As we already discussed, Murtagh and Vadhan~\cite{murtagh2016complexity} showed that finding the optimal values for the privacy parameters when composing different algorithms in a black-box way is  $\#\P$-complete, but also that approximating the optimal values can be done efficiently. In contrast, our results show that when one wants to consider programs as white-box, as often needed to achieve better privacy guarantees (e.g. in the case of the sparse vector technique), the complexity is higher.

Several works have explored different property testing related to differential privacy~\cite{DixitJRT13,JhaR13,gilbert2018property}, including verification~\cite{gilbert2018property}. In the standard model used in property testing, a user has only black-box access to the function and the observable outputs are the ones provided by a privacy mechanism. In contrast, our work is based on the program description and aim to provide computational limits to the design of techniques for program analyses for differential privacy.

We already discussed some works on quantitative information flow. In addition to those, it was shown that comparing the quantitative information flow of two programs on inputs coming from the uniform distribution is $\#\P$-hard \cite{yasuoka2010quantitative}. However, when quantifying over all distributions the question is $\coNP$-complete \cite{yasuoka2010quantitative}. 

As we remarked earlier, our language is equally expressive when integers of a fixed size are added. Recently Jacomme, Kremer and Barthe~\cite{barthe2020lics} show deciding equivalence of two such programs, operating over a fixed finite field, is $\coNP^{\C_=\P}$-complete and the majority problem, which is similar to pure differential privacy, is $\coNP^{\PP}$-complete---matching the class we show for deciding $\epsilon$-differential privacy. Further the universal equivalence problem, which shows the programs are equivalent over all field extensions, is decidable in 2-$\EXP$; the universal majority problem is not know to be decidable.

\section{Conclusions and future work}
\label{sec:conclusions}

\paragraph*{Verifying differential privacy of loop-free probabilistic boolean programs.} We have shown the difficulty of verifying differential privacy in loop-free probabilistic boolean programs through their correspondence with probabilistic circuits. Deciding  $\epsilon$-differential privacy  is $\conpsp{}$-complete and \eddpnoun{} is $\conpsp{}$-hard and in $\conpspsp{}$ (a gap that we leave for future work). Both problems are positioned in the counting hierarchy, in between the polynomial hierarchy $\PH$ and $\PSPACE$.

\paragraph*{Verifying differential privacy of probabilistic boolean programs.}
One interesting question that our work leaves open is the characterization of the complexity of deciding differential privacy problems for probabilistic boolean programs, including loops. Similarly to the works on quantitative information flow~\cite{chadha2014quantitative}, we expect these problems to be decidable and we expect them to be in $\PSPACE$. However, this question requires some further investigation that we leave for future work.

\paragraph*{Solvers mixing non-determinism and counting. }
Returning to our motivation for this work---developing practical tools for verifying differential privacy---our results seem to point to a deficiency in available tools for model checking. The model checking toolkit includes well established \textsc{Sat} solvers for $\NP$ (and $\coNP$) problems, solvers for further quantification in $\PH$, solvers for \textsc{\#Sat} (and hence for $\#\P$ problems\footnote{See, for example, \url{http://beyondnp.org/pages/solvers/}, for a range of solvers}). However to the best of our knowledge, there are currently no solvers that are specialized for mixing the polynomial hierarchy $\PH$ and counting problems $\#\P$, in particular $\conpsp{}$ and $\conpspsp{}$.

\paragraph*{Approximating the differential privacy parameters.} We show that distinguishing $(\epsilon,\delta)$-differential privacy from $(\epsilon',\delta')$ differential privacy where $(\epsilon,\delta) < (\epsilon',\delta')$ is both $\NP$- and $\coNP$-hard. We leave refining the classification of this problem as an open problem.

\newpage
\appendix

\section{Equivalence and expressivity of circuits and programs}\label{appen:integersandequivalence}
This section will prove \cref{lemma:circuiteqprogram}, but first let us expand on \cref{remark:integers} to prove \cref{lemma:circuiteqprogram} for an apparently more general language. We observe in \cref{remark:equivwithwithoutintegers} that they are equivalent.

Let us consider the following consider programs expressed by the following language. This is an extension to the language defined in  \cref{sec:lfpp}, supporting integer operations.
\[\begin{array}{rcl@{\qquad}r}
x&::=& [\text{a}{-}\text{z}]^+ & \text{(variable identifiers)} \\
i&::=&  0 \ | \ 1  \ | \ \dots  \ | \ 2^{k}-1  & \text{(constants)} \\
b&::=& {\tt true}\ |\ {\tt false}\ |\ {\tt random}\ |\ x\ |\  b \wedge b\ |\ b \vee b\ |\  \neg b\  & \text{(boolean expressions)} \\
& & \qquad | \  e> e\ |\ e\ge e\ |\ e=e\\
e&::=&  i\ |\ x\ |\ e+e\ |\ e\times e\ |\ e-e\ |\ (e) \ |\ b\ |\ \texttt{sample} & \text{(integer expressions)} \\
c&::=& {\tt SKIP}\ |\ x := e\ |\ c;\ c\ |\ {\tt if}\ b\ {\tt then}\ c\ {\tt else}\ c & \text{(commands)} 
\\t&::=& x\ |\ t, x & \text{(list of variables)} \\
p&::=& {\tt input}(t);\  c;\ {\tt return}(t) & \text{(programs)} 
\end{array}
\]

We define the programming language to operate over integers of a fixed size, parametrised by $k$.  Integer expressions, denoted $e$, take on values from the field $F_{2^k} = \{0,\dots,2^{k}-1\}$, where each integer can be represented with $k$ bits. Formally we consider the operations working on unsigned integers, where overflows wrap around. Similarly one could equally consider signed integers operating over $\{-2^{k-1},\dots,2^{k-1}-1\}$. Boolean values can also be interpreted as the integers $0$ and $1$. The \texttt{sample} command uniformly chooses from $\{0,\dots,2^{k}-1\}$; this  is syntactic sugar for $2^{k-1} \times \texttt{random} + \dots + 2^0 \times \texttt{random}$. 

\begin{remark}\label{remark:equivwithwithoutintegers}This language is equally expressive without integers, as Boolean operations can be used to simulated them---although programming in such a language may be somewhat tedious. To see this note that the following proof shows that programs in these extended language can be converted into probabilistic circuits, and similarly converted back into the Boolean fragment of the program. Both of these conversions occur in \textit{linear} time.\end{remark}

\begin{proof}[Proof of \cref{lemma:circuiteqprogram}]
As noted in the proof sketch, it is clear that probabilistic boolean circuits can be converted into probabilistic loop-free Boolean programs.

For the reverse conversion, we show any probabilistic loop-free  program \textit{with Booleans and bounded integers} can be converted into a probabilistic \textit{Boolean circuit} by showing that every expression can be constructed as a fixed sub-circuit.
\begin{itemize}

\item Boolean operations $\wedge,\vee, \neg$ are built into the circuit.

\item  \texttt{random} is achieved by introducing a fresh random input to the circuit for each call to $\texttt{random}$.

\item In the evaluation of expressions $e$, the value of each sub-expression can be stored as the output at most $k$ gates. The standard operations ${+},{\times},{-},{\ge},{>},{=}$ can each be computed in a circuit (see e.g. \cite{oberman1979digital}). 
Multiplication of two $k$-bit numbers may require $2k$-bits; since the language is defined over the field, these are renormalised back into the field by taking the $k$ least significant bits. Similarly adding two $k$-bit numbers results in $k+1$-bits, which must similarly be transformed back into the field. Note that it is acceptable to have up to $2k$ bits temporarily in the circuit.

\item Assignment $x := e$. Denote the $k$ gates as the output of expression $e$ as the variable $x$. References to $x$ then use these $k$ output bits.

\item Branching ${\tt if}\ b\ {\tt then}\ c_L\ {\tt else}\ c_R$ simulated in the following way: Create a circuit for each branch, taking two copies of each variable that can be assigned in either $c_L$ or $c_R$. At the end the true copy of each variable can be set. Let $(x)_i$ be the $i^\text{th}$ bit of variable $x$  then set $x$ by $(x)_i \leftarrow ((x_L)_i \wedge b) \vee ((x_R)_i \wedge (\neg b))$, so that $x = x_L$ if $b$ is true and $x = x_R$ if $b$ is false.

\item Constants are achieved by expressing the number in binary and using $k$ gates to represent the value. Any gate can be forced to be one by $x \vee \neg x$ and forced to be zero by $x\vee\neg x$ for any input or random bit.

\item ${\tt return}(x_1,\dots,x_m)$ is achieved by dedicating the $m \times k$ gates as output gates.\qedhere

\end{itemize}

\end{proof}

\begin{example}\label{eg:rr0.25}
\cref{fig:eg:rr0.25} and the following loop-free probabilistic program demonstrates an equivalent implementation of randomised response with $p=0.25$, entirely within the Boolean fragment.
\begin{lstlisting}[mathescape=true]
input($x$) 
$r$ := random $\wedge$ random 
$y$ := $(x \wedge \neg r) \vee (\neg x \wedge r) $
return($y$)
\end{lstlisting}

\begin{figure}
\centering
\includegraphics[width=0.8\linewidth]{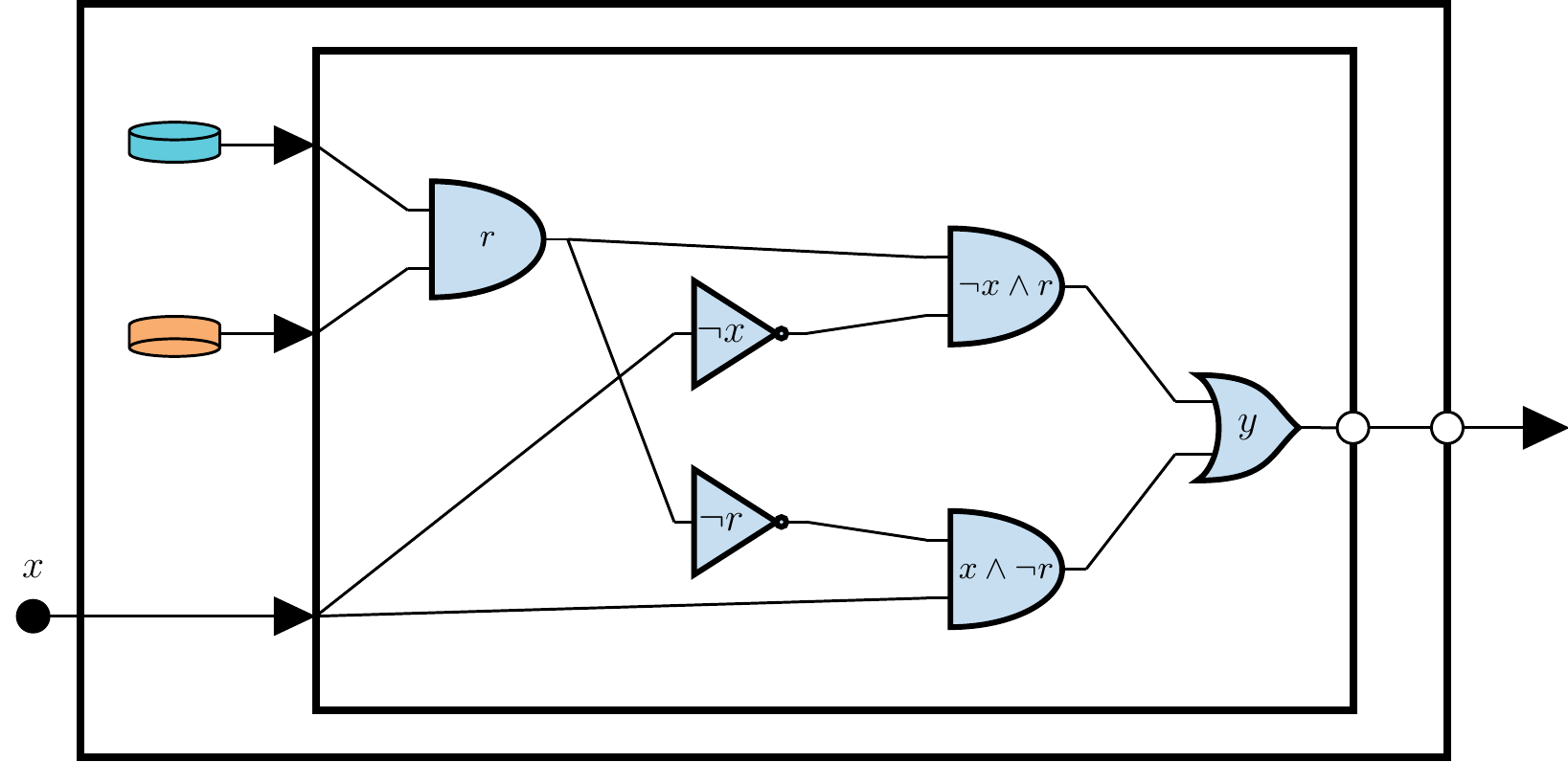}
\caption{Example Circuit Transformation from \cref{eg:rr0.25}}
\label{fig:eg:rr0.25}
\end{figure}

\end{example}

\begin{example}\label{eg:egtrans}
Consider the following loop-free probabilistic program, with $k=4$. It can be transformed into the circuit shown in \cref{fig:egtrans} through the procedure given in the proof of \cref{lemma:circuiteqprogram}.
\begin{lstlisting}[mathescape=true]
input(a,b)
x = 3 $\times$ a
y : = x + b
if random then z := x else z := y
return(z)
\end{lstlisting}

\begin{figure}
\centering
\includegraphics[width=0.9\textwidth]{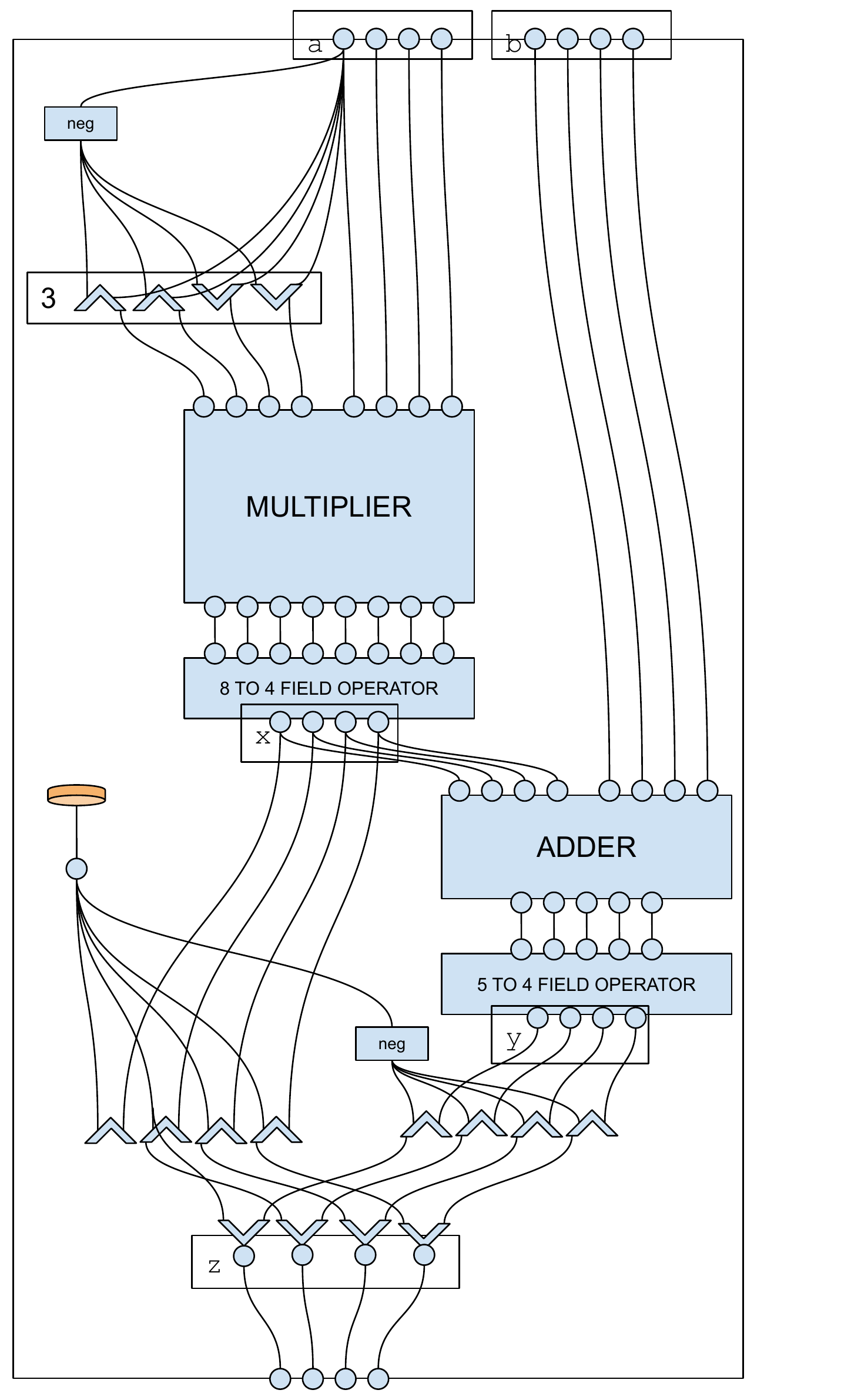}
\caption{Example Circuit Transformation from \cref{eg:egtrans}}
\label{fig:egtrans}
\end{figure}

\end{example}

\section{Hardness of \decideedp{} by number of input/output bits}\label{appen:circuitshape}

\begin{lemma}\label{lemma:pphard}
Given a circuit $\generalcircuit{}$, we show the the following hardness results for large and small number of input and output bits:

\begin{center}
\begin{tabular}{|@{\hskip 0.2cm}c @{\hskip 0.2cm}|@{\hskip 0.2cm}c@{\hskip 0.2cm}|@{\hskip 0.2cm}l@{\hskip 0.2cm}|}
\hline
$\#$ Input Bits & $\#$ Output Bits & Hardness\\
\hline
$\Omega(n)$ & 1 & $\coNP^{\#\P}$-hard\\
\hline
1 & $\Omega(n)$ & $\coNP^{\#\P}$-hard\\
\hline
1 & 1 & $\PP$-hard \\
\hline
\end{tabular}
\end{center}
\end{lemma}

\begin{remark}
Note that the hardness results entail hardness for any larger number of input and output bits; for example $\Theta(\log n)$-input,$\Theta(\log n)$-output is $\PP$-hard and $\Theta(n)$-input,$\Theta(n)$-output is $\coNP^{\#\P}$-hard. 
\end{remark}

\begin{proof}[Proof for large input small output]
Given $\phi(\xa,\ya)$, we reduce $\phi \in \ams$ to $\decideedp{}$. Our resulting circuit $\psi$ will have $1$ output bit but $n + 1$ input bits

Let $\psi(\xa,z) = (z \vee p_1) \wedge ( \neg z  \vee (p_2 \vee (p_3 \wedge p_4 \wedge \phi(\xa,\proballoc{}))))$, with $p_1,\dots,p_4, \proballoc{}$ determined randomly. This circuit has the property:
\begin{itemize}
	\item If $z = 0$ return 1 w.p. $\frac{1}{2}$.
	\item If $z = 1$ return 1 w.p. $\frac{1}{2} + \frac{1}{4} \pr[\phi(\xa) = 1]$
\end{itemize}

\begin{claim} $\phi \in \ams \iff \ln(\frac{4}{3})$-differential privacy holds.\end{claim}
If $\phi \not\in \ams$ then for some $\xa$ with $\pr[\phi(\xa) = 1] > \frac{1}{2}$, $   \pr[\phi(\xa) = 0] < \frac{1}{2 }$
\begin{align*}
\frac{\pr[\psi(\xa, 0) = 0]}{\pr[\psi(\xa, 1) = 0]} &=\frac{\frac{1}{2}}{1 - \frac{1}{2} - \frac{1}{4}\pr[\phi(\xa) = 1])} 
\\&=\frac{\frac{1}{2}}{ \frac{1}{4}  + \frac{1}{4} - \frac{1}{4}\pr[\phi(\xa) = 1])}  
\\&==\frac{\frac{1}{2}}{ \frac{1}{4}  + \frac{1}{4} (1 -\pr[\phi(\xa) = 1])}
\\&=\frac{\frac{1}{2}}{ \frac{1}{4}  + \frac{1}{4} (1 -\pr[\phi(\xa) = 1])}
\\&=\frac{\frac{1}{2}}{ \frac{1}{4}  + \frac{1}{4} (\pr[\phi(\xa) = 0])}  
\\&> \frac{\frac{1}{2}}{ \frac{1}{4}  + \frac{1}{4} \frac{1}{2}} = \frac{4}{3} \approx 1.3
\end{align*}

If $\phi \in \ams$ then for all $\xa$ we have $\pr[\phi(\xa) = 1] \le \frac{1}{2}$, $   \pr[\phi(\xa) = 0] \ge \frac{1}{2}$
\[
\frac{\pr[\psi(\xa, 0) = 0]}{\pr[\psi(\xa, 1) = 0]} \le \frac{4}{3} \approx 1.3
\]
\begin{align*}
\frac{\pr[\psi(\xa, 1) = 0]}{\pr[\psi(\xa, 0) = 0]} &= \frac{1 - \frac{1}{2} - \frac{1}{4} \pr[\phi(\xa) = 1]}{\frac{1}{2}}
\\&= \frac{\frac{1}{4} + \frac{1}{4} \pr[\phi(\xa) = 0]}{\frac{1}{2}} \le 1
\end{align*}
\[
\frac{\pr[\psi(\xa, 1) = 1]}{\pr[\psi(\xa, 0) = 1]}  = \frac{\frac{1}{2} + \frac{1}{4}\pr[\phi(\xa) = 1] }{\frac{1}{2}} \le \frac{\frac{1}{2} + \frac{1}{4}\frac{1}{2} }{\frac{1}{2}} =1.25
\]
\[
\frac{\pr[\psi(\xa, 0) = 1]}{\pr[\psi(\xa, 1) = 1]}  = \frac{\frac{1}{2}}{\frac{1}{2} + \frac{1}{4}\pr[\phi(\xa) = 1] } \le 1
\]
\[
\frac{\pr[\psi(\xa, 0) = 1]}{\pr[\psi(\xa', 0) = 1]}  =\frac{\frac{1}{2}}{\frac{1}{2}} = 1
\]
\[
\frac{\pr[\psi(\xa, 1) = 1]}{\pr[\psi(\xa', 1) = 1]}  =\frac{\frac{1}{2} + \frac{1}{4} \pr[\phi(\xa) = 1] }{\frac{1}{2} + \frac{1}{4} \pr[\phi(\xa') = 1] } \le \frac{\frac{1}{2} + \frac{1}{4}\frac{1}{2} }{\frac{1}{2} + \frac{1}{4} 0 }= 1.25
\]
\[
\frac{\pr[\psi(\xa, 1) = 0]}{\pr[\psi(\xa', 1) = 0]}  =\frac{1 - \frac{1}{2} - \frac{1}{4} \pr[\phi(\xa) = 1] }{1- \frac{1}{2} - \frac{1}{4} \pr[\phi(\xa') = 1] } \le \frac{\frac{1}{2} + \frac{1}{4} 0 }{\frac{1}{2} - \frac{1}{4}\frac{1}{2} } = \frac{4}{3}\qedhere
\]\end{proof}

\begin{proof}[Proof for small input large output]
Given $\phi(\xa,\ya)$, we reduce $\phi \in \ams$ to
$\decideedp{}$. Our resulting circuit $\psi$ will have $1$ input bit
but $n + 1$ output bits.

Let $\psi(z) =(\xa, (z \vee p_1) \wedge ( \neg z  \vee (p_2 \vee (p_3 \wedge p_4 \wedge \phi(\xa,\proballoc{})))))$, with $p_1,\dots,p_4,\xa{},\proballoc{}$ all chosen randomly. Then the circuit has the property:
\begin{itemize}
	\item Choose and output some $\xa$ and,
	\item If $z = 0$ return 1 w.p. $\frac{1}{2}$.
	\item If $z = 1$ return 1 w.p. $\frac{1}{2} + \frac{1}{4} \pr[\phi(\xa) = 1]$
\end{itemize}

\begin{claim} $\phi \in \ams \iff \ln(\frac{4}{3})$-differential privacy holds.\end{claim}

If $\phi \not\in \ams$ then for some $\xa$ with $\pr[\phi(\xa) = 1] > \frac{1}{2}$, $   \pr[\phi(\xa) = 0] < \frac{1}{2 }$
\begin{align*}
\frac{\pr[\psi(0) = (\xa,0)]}{\pr[\psi(1) = (\xa,0)]} &=\frac{\frac{1}{2}}{1 - \frac{1}{2} - \frac{1}{4}\pr[\phi(\xa) = 1])}
\\& =\frac{\frac{1}{2}}{ \frac{1}{4}  + \frac{1}{4} - \frac{1}{4}\pr[\phi(\xa) = 1])} 
\\& =\frac{\frac{1}{2}}{ \frac{1}{4}  + \frac{1}{4} (1 -\pr[\phi(\xa) = 1])}
\\& =\frac{\frac{1}{2}}{ \frac{1}{4}  + \frac{1}{4} (1 -\pr[\phi(\xa) = 1])}
\\& =\frac{\frac{1}{2}}{ \frac{1}{4}  + \frac{1}{4} (\pr[\phi(\xa) = 0])}  
\\& > \frac{\frac{1}{2}}{ \frac{1}{4}  + \frac{1}{4} \frac{1}{2}} = \frac{4}{3} \approx 1.3
\end{align*}

If $\phi \in \ams$ then for all $\xa$ we have $\pr[\phi(\xa) = 1] \le \frac{1}{2}$, $   \pr[\phi(\xa) = 0] \ge \frac{1}{2}$
\[
\frac{\pr[\psi(0) = (\xa,0)]}{\pr[\psi(1) = (\xa,0)]} \le \frac{4}{3} \approx 1.3
\]
\begin{align*}
\frac{\pr[\psi(1) = (\xa,0)]}{\pr[\psi(0) = (\xa,0)]} &= \frac{1 - \frac{1}{2} - \frac{1}{4} \pr[\phi(\xa) = 1]}{\frac{1}{2}}
\\& = \frac{\frac{1}{4} + \frac{1}{4} \pr[\phi(\xa) = 0]}{\frac{1}{2}} \le 1
\end{align*}
\[
\frac{\pr[\psi(1) = (\xa,1)]}{\pr[\psi(0) = (\xa,1)]}  = \frac{\frac{1}{2} + \frac{1}{4}\pr[\phi(\xa) = 1] }{\frac{1}{2}} \le \frac{\frac{1}{2} + \frac{1}{4}\frac{1}{2} }{\frac{1}{2}} =1.25
\]
\[
\frac{\pr[\psi(0) = (\xa,1)]}{\pr[\psi(1) = (\xa,1)]}  = \frac{\frac{1}{2}}{\frac{1}{2} + \frac{1}{4}\pr[\phi(\xa) = 1] } \le 1\qedhere
\]
\end{proof}

\begin{proof}[Proof for small input small output]
Given $\phi(\xa)$, we reduce $\phi \in \textsc{Maj-Sat}$ to $\decideedp{}$. Our resulting circuit $\psi$ will have $1$ output bit and $1$ input bits

Let $\psi(z) = (p_1 \wedge z) \vee (\neg p_1 \wedge (z \oplus \phi(\proballoc{})))$, where $p_1$ and $\proballoc{}$ are chosen randomly.
Then the circuit has the property:
\begin{itemize}
	\item probability $\frac{1}{2}$ output $z$. 
	\item probability $\frac{1}{2}$ output $z \oplus \phi(\proballoc{})$. (Output $z$, flipped proportionally to the number of accepting allocations to $\phi$.)
\end{itemize}

\begin{claim} $\phi \in \textsc{Maj-Sat} \iff \psi $ is $\ln(3)$-differentially private\end{claim}

The probabilities behave as follows, where each case is also bound by $\frac{1}{2}$ in the direction consistent with the probability shown.

\begin{tabular}{l|l|p{1cm}p{1cm} |p{2cm}}
Output $\downarrow$             & Input $\rightarrow$ & 1 & 0 & Max-Ratio  \\ \hline
\multirow{2}{*}{0} & \textsc{Maj}    & $>\frac{1}{4}$  & $<\frac{3}{4}$ &$ >3$ \\
                   & \textsc{Min}    & $\le\frac{1}{4}$  & $\ge\frac{3}{4}$ &$ \le3$ \\
\multirow{2}{*}{1} & \textsc{Maj}    & $<\frac{3}{4}$  & $>\frac{1}{4}$ &$ >3$ \\
                   & \textsc{Min}    & $\ge\frac{3}{4}$  & $\le\frac{1}{4}$ &$ \le3$ \\
\end{tabular}

Then in the either \textsc{Maj} case we have the ratio is greater than $3$ (violating privacy) and for both \textsc{Min} case the ratio is bounded by 3 (satisfying privacy).
\end{proof}

\section{Direct Proof that \decideeddp{} is \conpsp{}-hard}\label{appen:directhardnessproof}

We prove that $\decideeddp{}$ is $\conpsp{}$-hard, even when there is just one output bit and for every $\epsilon$.
\begin{theorem} \label{thm:decideeddpisconpsphard}
\decideeddp{} is \conpsp{}-hard.
\end{theorem}

We could show  $\decideeddp{}$  by reduction from $\ams$, which would entail that $(\epsilon,\frac{1}{2})$-differential privacy is $\conpsp{}$-hard. To show hardness for a large range of $\delta$ we first generalise $\ams$ to $\afs$, showing this is also hard. We will then reduce $\afs$ to $\ams$.

\subsection{Generalising \ams}

Let us first generalise $\ams$ to \afs, which rather than requiring that the minority (half) of allocations to $y$ give true, rather no more than a fraction $f$. Similarly we generalise $\ems$ to \efs.

\begin{definition}
A formula $\phi(\xa,\ya)$, $\xa\in\zo{n},\ya\in\zo{m}$ and $f \in [0,1] \cap \mathbb{Q}$ is \afs if for every $\xa\in\zo{n}$
$$
\frac{|\{\ya \in \zo{m} \ | \ \phi(\xa,\ya) \text{ is true}\}|}{2^m} \le f
$$

\end{definition}
\begin{remark}
For $f = 0$, we require that for all $\xa$, no input of $\ya$ gives true, therefore we require $\phi$ is unsatisfiable. For $f = 1$, we have essentially no restriction and for $f = \frac{2^m - 1}{2^m}$ we require that $\phi$ is not a tautology. $\ams{}$ is then when $f = \frac{1}{2}$.
\end{remark}

\begin{definition}A formula $\phi$ is $\efs$ if it is not $\afs$. 
\end{definition}

This means a formula $\phi(\xa,\ya)$, $\xa\in\zo{n},\ya\in\zo{n}$ and $f \in [0,1] \cap \mathbb{Q}$ is \efs if there exists an allocation $\xa$ that more than $f$ fraction of allocations to $\ya$ result in $\phi(\xa,\ya)$ being true. That is there exists $\xa\in\zo{n}$ such that
$
\frac{|\{\ya \in \zo{m} \ | \ \phi(\xa,\ya) \text{ is true}\}|}{2^m} > f.
$

Towards showing $\afs$ is $\coNP^{\#\P}$-hard, we show \efs is $\NP^{\#\P}$-hard, entailing \cref{lem:afsconpsp}.

\begin{lemma}\label{lem:efsnpsp}\efs is $\NP^{\#\P}$-hard for $f\in[\frac{1}{2^m},\frac{2^m - 1}{2^m})$.
\end{lemma}

\begin{corollary}\label{lem:afsconpsp}\afs is $\coNP^{\#\P}$-hard for $f\in[\frac{1}{2^m},\frac{2^m - 1}{2^m})$.
\end{corollary}

\begin{proof}[Proof of \cref{lem:efsnpsp} for $f$ of the form $\frac{1}{2^{k+1}}$]

We reduce  \ems, given  a formula $\phi(\xa,\ya)$ to \efps{$\frac{1}{2^{k+1}}$}.

(We assume $\frac{1}{2^{k+1}}$ takes $O(k)$ bits.)

Define a formula $\phi'(\xa,\genalloc{w}{})$, with $\genalloc{w}{} \in \zo{m + k}$. Let $\genalloc{w}{} = (y_1, \dots, y_m, z_1 \dots z_k)$, where $\ya = (y_1, \dots, y_m)$.

$\phi'(\xa,\genalloc{w}{})=  z_1 \wedge \dots \wedge z_k \wedge \phi(\xa, y_1, \dots, y_m)$

For $\xa $ fixed if $g$ allocations to $y_1, \dots, y_m$ satisfy $\phi$ then each of these satisfy $\phi'$ only when $z_1 = \dots = z_k = 1$. All remaining times are unsatisfied. 

Suppose $\frac{g}{2^m}$ of $\ya$'s satisfy $\phi(\xa,\ya)$ then $\frac{g}{2^m \cdot 2^k}$ $\genalloc{w}{}$'s satisfy $\phi'(\xa,\genalloc{w}{})$. 

That is we have:
\[\frac{g}{2^m \cdot 2^k} > \frac{1}{2^{k+1}} \iff \frac{g}{2^m} > \frac{1}{2}.\qedhere\]
\end{proof}

\begin{proof}[Proof of \cref{lem:efsnpsp} for $f$ of the form $\frac{a}{2^{k+1}}$]
We reduce  \ems, given  a formula $\phi(\xa,\ya)$ to \efps{$\frac{1}{2^{k+1}}$}. Assume $a$ odd (otherwise, half and take $\frac{a / 2}{2^k}$) and greater than 1 (otherwise use above).

Define a formula $\phi'(\xa,\genalloc{w}{})$, with $\genalloc{w}{} \in \zo{m + k}$. Let $\genalloc{w}{} = (y_1, \dots, y_m, z_1 \dots z_k)$, where $\ya = (y_1, \dots, y_m)$.  Let $b = \frac{a-1}{2}$ (b is always between 1 and $2^k - 1$).

Let $\chi_\frac{b}{2^k}(z_1, \dots, z_k))$ be a circuit on $k$ bits, which is true for $\frac{b}{2^k}$ of its inputs of $z_1 \dots z_k$, but not true for $z_1 = \dots = z_k = 1$ when $b < 2^k$. \footnote{Given $b,k$ such that $ 0  < \frac{b}{2^k} < 1$, we create a formula, over $2k$ bits, which given two $k$-bit integers $m,n$, return whether $m \le n$ (such a formula is of size polynomial in $k$). By fixing $n$ to $b$, we have a circuit on $m$ input bits which decides if $m \le b$. Instead sample over the $m$ bits of $m$ producing a circuit $\chi_\frac{b}{2^k}$ which is true on $\frac{b}{2^k}$ of its inputs. }

Then we let $\phi'(\xa,\genalloc{w}{})= (z_1 \wedge \dots \wedge z_k \wedge \phi(\xa, y_1, \dots, y_m) )\vee \chi_{\frac{b}{2^k}}(z_1, \dots, z_k)$.

That is the formula $\phi'$ is true whenever $z_1 = \dots = z_k = 1$ and $\phi$ is true, or on the $\frac{b}{2^k}$ choices of $z_1 \dots z_k$.

Suppose $\frac{g}{2^m}$ of $\ya$'s satisfy $\phi(\xa,\ya)$ then $\frac{g}{2^m \cdot 2^k} + \frac{b}{2^k} =\frac{g}{2^m \cdot 2^k} + \frac{a-1}{2^{k+1}}$ $\genalloc{w}{}$'s satisfy $\phi'(\xa,\genalloc{w}{})$. 

That is we have:
\[\frac{g}{2^m \cdot 2^k} + \frac{a-1}{2^{k+1}} > \frac{a}{2^{k+1}} \iff \frac{g}{2^m \cdot 2^k} > \frac{1}{2^{k+1}}\iff\frac{g}{2^m} > \frac{1}{2}\qedhere\]

\end{proof}

\begin{proof}[Proof of \cref{lem:efsnpsp} for $f$ of the form $\frac{a}{b}$]

Let $m$ be the number such that $\ya \in\zo{m}$, the number of $\ya$ bits of the formula, or the number of `MAJ' bits. Let $z$ be such that $\frac{z}{2^m} < \frac{a}{b} < \frac{z+1}{2^m}$. We reduce \efps{$\frac{z}{2^m}$} to \efps{$\frac{a}{b}$}, by simply taking $\phi$ unchanged.

Suppose $\frac{g}{2^m}$ of $\ya$'s satisfy $\phi(\xa,\ya)$ then 
\[\frac{g}{2^m}> \frac{z}{2^m} \iff \frac{g}{2^m} \ge \frac{z+1}{2^m}\iff \frac{g}{2^m} > \frac{a}{b}.\qedhere\]\end{proof}

\subsection{Main Proof of \texorpdfstring{\cref{thm:decideeddpisconpsphard}}{Theorem~\ref{thm:decideeddpisconpsphard}}}

\begin{proof}
Assume we are given an instance of \afs, a formula $\phi(\xa,\ya)$ for $\xa \in \zo{n}, \ya \in \zo{m}$ and $f \in [0,1]$.  

We define a circuit $\psi$, with inputs $\inp \in\zo{n+1}$, we write as $(z, \xa{}_1,\dots, \xa{}_n)$; matching the inputs $\xa$ and an additional bit $z$. There are $m$ probabilistic bits $\proballoc{} \in\zo{m}$, matching $\ya$. There is one output bit $\outpt{} \in \zo{1}$.  The circuit $\psi$  will behave like $\phi$ when $z = 1$ and simply output $0$ when $z = 0$; i.e. $\outpt{}_1 = z \wedge \phi(\xa{}_1,\dots, \xa{}_n,\proballoc{}_1,\dots, \proballoc{}_m)$.

\begin{claim} $\phi \in$ \afs if and only if $\psi$ is $(\epsilon,\delta)$-differentially private, for $\delta = f$ and any choice of $\epsilon$ (including zero).
\end{claim}

\paragraph*{Direction: if $\phi \not\in$ \afs then not $(\epsilon,\delta)$-differentially private.}

Given $\phi \not\in$ \ams then there exists $\xa{}\in\zo{n}$ such that $\phi(\xa{},\ya{})$ is on more than $f$ portion of $\ya{}\in\zo{m}$. We show the differential privacy condition is violated exactly using this $\xa{}$, let $\inp = (1, \xa{}_1, \dots, \xa{}_n)$ and $\inpp = (0, \xa{}_1, \dots, \xa{}_n)$. Let us consider the probability of the event $o_1 = 1$.

Then we have $\pr[\psi(1, \xa{}_1, \dots, \xa{}_n) = 1] > f$ and $\pr[\psi(0, \xa{}_1, \dots, \xa{}_n)  = 1] = 0$. Violating differential privacy since, \begin{equation*}\pr[\psi(1, \xa{}_1, \dots, \xa{}_n) = 1] - e^\epsilon \pr[\psi(0, \xa{}_1, \dots, \xa{}_n) = 1] 
 > f - 0 =  \delta .\end{equation*}

\paragraph*{Direction: if $\phi \in$\afs then $(\epsilon,\delta)$-differentially private.}

Since $\phi \in$ \afs then for all $\xa{}\in\zo{n}$ we have $\phi(\xa,\ya)$ true for less or equal $f$ proportion of the allocations to $\ya{}\in\zo{m}$. Equivalently the more than $f$ of $\ya\in\zo{m}$ with $\phi(\xa,\ya)$ false.

To show privacy we consider all adajcent inputs and all output event. The output events are $E \subseteq\zo{1}$, giving $\{\}, \{0\},\{1\},\{0,1\}$. The probability of `no output' $\{\}$ is zero for all inputs, so cannot violate differential privacy. The probability of `output anything' $\{0,1\}$ is one for all inputs, so does cannot violate differential privacy. Thus we argue that events $\{0\}$ and $\{1\}$ do not violate differential privacy, for all adjacent inputs. 

Inputs take the form $\inp = (z, \xa{}_1, \dots, \xa{}_n)$, and $\inp,\inpp$ can be adjacent either with fixed $\xa$ and differing $z$ or, fixed $z$ and $\xa$ differing in one position; in each case we show $\pr[\psi(\inp) = E] -e^\epsilon \pr[\psi(\inpp) = E] \le \delta$ and $\pr[\psi(\inpp) = E] -e^\epsilon \pr[\psi(\inp) = E] \le \delta$

\paragraph*{Let $z$ be fixed to zero.} Hence we have $\inp,\inpp$ with $\xa{}$'s differing in one position. For $z = 0$ the circuit outputs zero in all cases, independently of $\xa{}$, thus does not violate differential privacy since $\pr[\psi(\inp) = E] = \pr[\psi(\inpp) = E]$.

\paragraph*{Let $z$ be fixed to one.} Hence we have $\inp,\inpp$ with $\xa{}$'s differing in one position. Without loss of generality suppose the difference is $x_j$.

For the event $E = \{1\}$ we have the probability being $\le f$ for each input; that is regardless of $\xa{}$ we have $\pr[\psi(\inp)  = 1] \le f$. So $\pr[\psi(\inp) = E] - e^\epsilon \pr[\psi(\inpp) = E] \le \pr[\psi(\inp) = E]  \le f = \delta$ and $\pr[\psi(\inpp) = E] - e^\epsilon pr[\psi(\inp) = E] \le \pr[\psi(\inpp) = E]\le f =\delta$.

For the event $E = \{0\}$ we have the probability being $\ge 1 - f$ for each input; that is regardless of $\xa{}$ we have $\pr[\psi(\inp)  = 0] \ge 1 - f$. So $\pr[\psi(\inp) = E] - e^\epsilon \pr[\psi(\inpp) = E] \le 1 - \pr[\psi(\inpp) = E]  \le f = \delta$ and $\pr[\psi(\inpp) = E] - e^\epsilon \pr[\psi(\inp) = E] \le 1- \pr[\psi(\inp) = E]\le f =\delta$.

\paragraph*{Let $\xa{}_1, \dots, \xa{}_n$ be fixed.}
We have $\xa{}_1, \dots, \xa{}_n$ fixed and the case $z = 0 $ and $z=1$, hence we have $\inp = (1,\xa{}_1, \dots, \xa{}_n)$ and  $\inpp = (0,\xa{}_1, \dots, \xa{}_n)$.

For the event $\{1\}$, when $z = 1$, we have $\pr[\psi(\inp) = 1] \le f$, but for $z=0$ the circuit is always $0$, thus $\pr[\psi(\inpp) = 1] = 0$.

Then $\pr[\psi(\inp) = 1] - e^\epsilon \pr[\psi(\inpp) = 1]= \pr[\psi(\inp) = 1] \le f =\delta$ and $\pr[\psi(\inpp) = 1] - e^\epsilon \pr[\psi(\inp) = 1] \le 0 \le \delta$.

For the event $\{0\}$ we have then $\pr[\psi(\inp) = 0]  \ge 1 - f$ and $\pr[\psi(\inpp)  = 0] = 1$
Then $\pr[\psi(\inp) = 0] - e^\epsilon \pr[\psi(\inpp) = 0]\le 1 - e^\epsilon \le 0 \le \delta$ and $\pr[\psi(\inpp) = 0] - e^\epsilon \pr[\psi(\inp) = 0]\le 1 - \pr[\psi(\inp) = 0] \le f = \delta$.\end{proof}

\section{Conditioning}\label{appen:conditioning}  Conditioning allows the run of a program to fail, so that the probability associated with failing runs is renormalised over all other runs (see e.g. \cite{DBLP:journals/toplas/OlmedoGJKKM18}). We show our decision procedures are robust to this notion, for which we simulate failure with an additional output bit and incorporate the renationalisation into our decision procedures. Naturally our lower bounds apply to this more general notion.

Conditioning can be encoded in a circuit by assuming a bit  which indicates whether the run has succeeded; in the case of failure we can assume all other output bits are false (denoted $\genalloc{0}$).  Thus for `proper' events in $\zo{\ell}$, the circuit formally outputs from  $\zo{\ell+1}$. We redefine our probability of an event as  \[
\pr[\generalcircuit{}(\xa{}) \in E] = \frac{|\{\proballoc{} \in \zo{m} \ | \ \generalcircuit{}(\xa, \proballoc{}) = (\top, \outpt{}) \text{ with } \outpt{}\in E\}|}
{|\{\proballoc{} \in \zo{m} \ | \ \generalcircuit{}(\xa, \proballoc{}) \ne (\bot, \genalloc{0})\}|}.
\]
Note that $|\{\proballoc{} \in \zo{m} \ | \ \generalcircuit{}(\xa, \proballoc{}) \ne (\bot, \genalloc{0})\}|$ is independent of the choice of $E$.

\subsection{\decideedp{}}
We generalise the procedure for $\decideedp{}$, maintaining $\conpsp{}$.
\begin{enumerate}
\item Guess $\xa,\xa' \in\zo{n},\outpt{}\in\zo{l}$
\begin{enumerate}
\item Compute $a = |\{\proballoc{} \in \zo{m} \ | \ \generalcircuit{}(\xa, \proballoc{}) = (\top, \outpt{})\}|$
\item Compute $b = |\{\proballoc{} \in \zo{m} \ | \ \generalcircuit{}(\xa', \proballoc{}) = (\top, \outpt{})\}|$
\item Compute $D_1 = |\{\proballoc{} \in \zo{m} \ | \ \generalcircuit{}(\xa, \proballoc{}) \ne (\bot, \genalloc{0})\}|$
\item Compute $D_2 = |\{\proballoc{} \in \zo{m} \ | \ \generalcircuit{}(\xa', \proballoc{}) \ne (\bot, \genalloc{0})\}|$
\item Reject if $D_1$ or $D_2$ is zero.
\item Accept if $a\cdot D_2 \le \exp(\epsilon)\cdot b \cdot D_2$ and otherwise reject.
\end{enumerate}
\end{enumerate}

\begin{claim*} $\generalcircuit{} \text{ is $\epsilon$-differentially private} \iff \text{$\decideedp{}$ accepts on all branches}$
\end{claim*}

\subsection{\decideeddp{}}
We generalise the procedure for $\decideeddp{}$, maintaining $\conpspsp{}$. Recall $e^\epsilon = \alpha = \frac{u}{v}$.
In more detail, on inputs $\generalcircuit$, $\xa$ and $\xa'$, $\mathcal{M}$ does the following:
\begin{enumerate}
  \item Choose $\outpt \in \zo{\ell}$ and an integer $C \in \{1,2,\dots, 2^{2m+\bits{v}}\}$ (this requires choosing $l+2m+\bits{v}$ bits). 
\item Through a call to the $\#\P$ oracle, $\mathcal{M}$ computes 
\begin{itemize}
\item $ a  =  \left|\{\proballoc{}\in\zo{m} : \psi(\xa,\proballoc{}) = (\top,\outpt)\}\right|$
\item $b  =  \left|\{\proballoc{}\in\zo{m} : \psi(\xa',\proballoc{}) = (\top,\outpt)\}\right|. $
\item  $D_1 = |\{\proballoc{} \in \zo{m} \ | \ \generalcircuit{}(\xa, \proballoc{}) \ne (\bot, \genalloc{0})\}|$
\item  $D_2 = |\{\proballoc{} \in \zo{m} \ | \ \generalcircuit{}(\xa', \proballoc{}) \ne (\bot, \genalloc{0})\}|$
\end{itemize}
\item 
$\mathcal{M}$ accepts if  $v \cdot a \cdot D_2 - u\cdot b\cdot D_1 \ge C$ and otherwise rejects.
\end{enumerate}
\begin{lemma} 
Given two inputs $\xa,\xa'\in X^n$, $\mathcal{M}(\psi,\xa,\xa'$) has exactly $v\cdot D_1 \cdot D_2  \sum_{ \outpt\in{\zo{\ell}}} \delta_{\xa,\xa'}(\outpt)$ accepting executions.
\end{lemma}

\begin{proof}
Let $\ind{X}$ be the indicator function, which is one if the predicate $X$ holds and zero otherwise.
\begin{align*}
\mkern-26mu&vD_1D_2  \sum_{\outpt \in \zo{\ell}}  \delta_{\xa,\xa'}(\outpt{})= \sum_{\outpt \in \zo{\ell}}v D_1 D_2 \max\left( \pr[\psi(\xa) = \outpt] - \alpha \pr[\psi(\xa') = \outpt], 0 \right) 
\\\mkern-26mu&= \sum_{\outpt \in \zo{\ell}}vD_1D_2 \max\left( \frac{1}{D_1} \mkern-6mu\sum_{\proballoc\in\zo{m}}\mkern-18mu\ind{\psi(\xa,\proballoc) = (\top,\outpt)} - \alpha \frac{1}{D_2}\mkern-6mu\sum_{\proballoc\in\zo{m}}\mkern-18mu\ind{\psi(\xa',\proballoc) = (\top,\outpt)}, 0 \right) 
\\\mkern-26mu&= \sum_{\outpt \in \zo{\ell}} \max\left( vD_2\mkern-18mu\sum_{\proballoc\in\zo{m}}\mkern-18mu\ind{\psi(\xa,\proballoc) = (\top,\outpt)} - v\alpha D_1\mkern-18mu\sum_{\proballoc\in\zo{m}}\mkern-18mu\ind{\psi(\xa',\proballoc) = (\top,\outpt)}, 0 \right) 
\\\mkern-26mu&= \sum_{\outpt \in \zo{\ell}} \max\left( vD_2\mkern-18mu\sum_{\proballoc\in\zo{m}}\mkern-18mu\ind{\psi(\xa,\proballoc) = (\top,\outpt)} - u D_1\mkern-18mu\sum_{\proballoc\in\zo{m}}\mkern-18mu\ind{\psi(\xa',\proballoc) = (\top,\outpt)}, 0 \right) 
\\\mkern-26mu&= \sum_{\outpt \in \zo{\ell}}\mkern-15mu \sum_{C= 1}^{2^{2m + \bits{v}}}\ind{\max\left( v D_2\mkern-18mu\sum_{\proballoc\in\zo{m}}\mkern-18mu\ind{\psi(\xa,\proballoc) = (\top,\outpt)} - u D_1 \mkern-18mu\sum_{\proballoc\in\zo{m}}\mkern-18mu\ind{\psi(\xa',\proballoc) = (\top,\outpt)}, 0 \right) \ge C}
\\\mkern-26mu&=\text{number of accepting executions in } \widehat{\mathcal{M}}\qedhere
\end{align*}\end{proof}

We can now describe our $\conpspsp{}$ procedure for $\decideeddp{}$. The procedure takes as input a probabilistic circuit $\generalcircuit$.
\begin{enumerate}
\item Non-deterministically choose neighboring $\xa$ and $\xa' \in \zo{n}$ (i.e.,  $2 n$ bits)
\item Let $\mathcal{M}$ be the non-deterministic Turing Machine with access to a $\#\P$-oracle as described above. Create a machine $\widehat{\mathcal{M}}$ with no input that executes $\mathcal{M}$ on $\generalcircuit, \xa,\xa'$ 
\item Make an $\#\P^{\#\P}$ oracle call  for the number of accepting executions $\widehat{\mathcal{M}}$ has. 
\item Make an $\#\P$ oracle call  for 
\begin{itemize}
	\item  $D_1 = |\{\proballoc{} \in \zo{m} \ | \ \generalcircuit{}(\xa, \proballoc{}) = (\bot, \genalloc{0})\}|$
\item  $D_2 = |\{\proballoc{} \in \zo{m} \ | \ \generalcircuit{}(\xa', \proballoc{}) = (\bot, \genalloc{0})\}|$
	\end{itemize}

\item Reject if $D_1$ or $D_2$ is zero.
\item Reject if the number of accepting executions of $\widehat{\mathcal{M}}$ is greater than $v \cdot  D_1 \cdot D_2 \cdot \delta$ and otherwise accept.
\end{enumerate}

\begin{claim*} $\generalcircuit{} \text{ is $(\epsilon,\delta)$-differentially private} \iff \text{$\decideeddp{}$ accepts on all branches} $
\end{claim*} 
\end{document}